\documentclass[12pt,reqno,fleqn]{article}

\usepackage[usenames]{color}
\usepackage{amssymb}
\usepackage{amsmath}
\usepackage{amsthm}
\usepackage{amsfonts}
\usepackage{amscd}
\usepackage{graphicx}
\usepackage{microtype}
\usepackage{hyperref}

\usepackage{fullpage}
\usepackage{float}

\usepackage{graphics}
\usepackage{graphicx}
\usepackage{latexsym}
\usepackage{epsf}

\usepackage{tikz}
\usepackage{tkz-tab}
\usetikzlibrary{shapes.geometric, arrows}
\tikzstyle{arrow} = [thick,->,>=stealth]
\usetikzlibrary{matrix}
\usetikzlibrary{calc}
\usetikzlibrary{fit}

\begin{document}

\newcommand{\seqnum}[1]{\href{https://oeis.org/#1}{\rm \underline{#1}}}

\theoremstyle{plain}
\newtheorem{theorem}{Theorem}
\newtheorem{corollary}[theorem]{Corollary}
\newtheorem{lemma}[theorem]{Lemma}
\newtheorem{proposition}[theorem]{Proposition}

\theoremstyle{definition}
\newtheorem{definition}[theorem]{Definition}
\newtheorem{example}[theorem]{Example}
\newtheorem{conjecture}[theorem]{Conjecture}

\theoremstyle{remark}
\newtheorem{remark}[theorem]{Remark}
\newcommand{\pf}{{\rm PF}}
\newcommand{\tp}{{\rm TP}}
\newcommand{\emd}{{\rm EMD}}
\def\R{\mathbb{R}}
\def\N{\mathbb{N}}

\begin{center}
\epsfxsize=4in
 
\end{center}

\begin{center}
\vskip 1cm{\LARGE\bf Runs, Squares, Palindromes, and Unbordered Factors of a Family of Binary Pattern Sequences with the All-One Pattern}
\vskip 1cm
\large
Russell Jay Hendel\\ 
Towson University\\
Towson, Maryland 21252\\
USA\\
\href{mailto:email}{RHendel@Towson.Edu}  
\end{center}

\vskip .2 in
\begin{abstract}
This paper presents results  on maximal runs, order of squares, palindromes, and unbordered factors of members of the family of binary pattern sequences with the all-one pattern. Restricting ourselves to binary pattern sequences with the all-one pattern with at least three ones,   five categories of maximal run lengths and 3 categories of orders of squares are presented,   palindromes with locally maximal length as well as palindromes with the second to fifth-largest palindrome lengths are described, and unbordered factors of lengths powers of two are presented. Interestingly, the characteristic functions of specified prefixes of sequences of the  2-kernel of these sequences can be formulated using the Vile and Jacobsthal sequences. Both Mathematica and Walnut are employed for exploratory pattern analysis. Proofs are based on a correspondence between binary strings under concatenation and integers under addition and multiplication. It is observed that this correspondence seems  most efficacious  for proofs of theorems whose statements are classified at low levels in the arithmetic hierarchy. 
\end{abstract}

\section{Introduction and Main Results}
The following conventions and notations, some of which are quite standard, are used throughout the paper. 

\begin{itemize}
\item  We let $\#,$ $|w|.$ $|w|_1, \bar{w},$ and  $xy$ or $x \cdot y$ respectively refer to cardinality,  the length of $w,$   the number of occurrences of the letter $1$ in the word $w,$ the binary complement of the binary string $w,$ and the concatenation of strings $x$ and $y.$  
\item We let  $aT+b = \{at_i +b\},$ where $T=\{t_i,i \in I\}$ is a set of integers indexed by $I$ and $a,b$ are integer constants. We let $s[I]=(s[i])_{i \in I}$ for $s$ a sequence and $I$ an indexing set.  
\item For non-negative integers, $a,b,$ $a^b$ represents concatenation if $a \in \{0,1\}$ and exponentiation otherwise.
\item For $n$ a binary string,   $n_v$ indicates its numerical value.  For $n$ a number, $n_2$ indicates its binary representation.  Except for $P,s,V,J,$ sequences defined below,  the paper does not use variables with subscripts, and hence, $n_2, i_2, j_2, x_2$ will have unambiguous meaning. 
\item When the meaning is unambiguous, we freely use arithmetic operations on binary strings, e.g. $1 \cdot 0 \cdot 1+1 \cdot 1 = 8$ is shorthand for $101_v + 11_v = 8,$ which in turn is shorthand for $(101)_v +(11)_v=8.$ Similarly, e.g., $1^m +1 = 2^m, m\ge 1,$ is shorthand for $ (1^m)_v + 1 = 2^m.$
\item  $e$ will refer to a variable over the non-negative even integers so that, e.g., $x \ge e+1$ means $x$ is a positive odd integer.
\item Although we let $s[b .. b+l-1] = s[b] s[b+1] \cdots s[b+l-1],$ with $s$ a sequence, $b$ a beginning position, and $l$ a factor length and $s[b \dotsc b+l-1] = (s[b+i])_{0 \le i \le l-1},$ 
we nevertheless, if there is no ambiguity, treat sequences as words and,  for example, refer to their prefixes, factors, and suffixes. Additionally, when convenient,  the paper  alternatively uses $s_{n}$ or $s[n]$ where $s$ is some sequence and similarly, for double sequences, $x_{m,n}$ and  $x_m[n]$ have the same meaning. 
\end{itemize}

 To motivate the object of study of this paper, recall that the $n$-th term of the Thue-Morse sequence, $n \ge 0,$ \seqnum{A010060},  is the parity of the number of ones occurring in the binary representation of 2. Similarly, the $n$-th term of the Rudin-Shapiro sequence, $n \ge 0,$ \seqnum{A020985},  is the parity of the number
 of occurrences of the binary string $11$ in the  binary representation of $n.$  
These facts immediately suggest the following natural generalization. 

\begin{definition} For $m \ge 1,$ define $P_m = 1^m.$ Define the family of sequences $(s_m)_{m \ge 1}$ by letting $s_{m,n}$ equal the parity of (possibly overlapping) occurrences of $P_m$ in $n_2,  n \ge 0.$      
\end{definition} 

 The $\{s_m\}_{m \ge 1}$ are the family of binary pattern sequences with the all-one pattern. Although their linear complexity  has been discussed \cite{Merai}, to this author's knowledge, the properties of their runs, squares, palindromes, and unbordered factors have not been studied. 

The main results of this paper relate the family of sequences $(s_m)_{m \ge 1}$ to certain basic concepts in automata theory.
Standard references for automata theory are \cite{Allouche, Rigo1, Rigo2, Shallit}. However, a problem arises in connection with
proof methods. Theorem \ref{the:main} proves that each $s_m$ is automatic. Hence, proofs about assertions of first order 
statements are decidable, and these days, can be accomplished by software \cite{Shallit, Mousavi}. However, the double sequence
$(s_{m,n})_{\{m \ge 1, n \ge 0\}}$ is not automatic (since, as shown in Theorem \ref{the:main}, the minimal DFAO accepting $s_m$ grows
unbounded as $m$ goes to infinity). Consequently, results about the sequence $(s_m)_{m \ge 1}$ require a proof method to establish
them. This paper employs a proof method exploiting a natural correspondence between integers under addition and multiplication and binary strings 
under concatenation.  The following Dictionary lemma, stated without proof (because the results are clear) presents basic facts 
about this correspondence. 

\begin{lemma}[Dictionary] \label{lem:Dictionary} 
\begin{enumerate}
\item[(i)]  For $e \ge 0,$ there are $e+1$ occurrences of $P_m$ in $P_{m+e}. $  
\item[(ii)]  For binary strings $x,y,$ $(x \cdot y)_v = x_v 2^{|y|} + y_v.$
\item[(iii)]
For $u,v,w,x \ge 0,$ $(1^u \cdot 0^v \cdot 1^w \cdot 0^x)_v  = (2^u-1) 2^{v+w+x} + (2^w-1)  2^{x}.$ 
\item[(iv)] For $j>m \ge 1, n \ge 1, (2^j n + 2^j - 2^m - 1)_2  = n_2 \cdot 1^{j-m-1} \cdot 0 \cdot 1^m,$
and $(2^j n +2^j-1)_2 = n_2 \cdot 1^j.$
\item[(v)]  In performing addition of binary representations of numbers, a useful technique, is to 
\textit{align exponents in the suffixes}, that is, rewrite the binary representations so that exponents
 of 0 and 1 align. For example, the addition $1^i \cdot 0^m + 1^{m-i}  \cdot  0^i$ is facilitated by 
rewriting it as $1^i \cdot 0^{m-i} \cdot 0^i + 1^{m-i} \cdot 0^i$ with two exponents aligned, facilitating
an immediate calculation of the numerical answer of $1^m \cdot 0^i.$ 
\end{enumerate}
\end{lemma}

In the sequel, it will be assumed, without being explicit,   
that any numerical computation on binary representations relies on this lemma. 

To avoid excessive narrative on technicalities associated with ranges, the Main Theorem and other results in this paper are formulated
for $m \ge 3.$ There is no loss of generality
in this, since the results for $m=1,2$ are well known and can both be discovered and proven by software.  
 
\begin{theorem}[Main]\label{the:main} For $m \ge 3,$ we have the following:
\begin{enumerate}
\item[(a)] There is a $2m$-state, deterministic, finite, automaton, with output (DFAO) accepting $s_m,$ whose transition function, $\delta,$ is given by \eqref{equ:delta}.
\begin{equation}\label{equ:delta}
\delta(q_i,a)= 
	\begin{cases} 
	\begin{aligned}
        q_0,			&\qquad  \text{ if } i \in \{0,\dotsc,m-1\}  					&		\text { and } a=0  \\ 
       q_{m+1}, 	&\qquad   \text{ if } i \in \{m,\dotsc, 2m-1\}   				&		\text { and } a=0   \\ 
	 q_{i+1},		&\qquad \text{ if } i \in \{0,\dotsc,m-1, m+1, \dotsc,  2m-2\} 	&		\text { and } a=1 \\
        q_{m-1},		&\qquad \text{ if } i=m	 								&		\text { and } a=1\\	 
        q_m,		&\qquad \text{ if }  i=2m-1  								&		\text { and } a =1.   
	\end{aligned}
	\end{cases}    
 \end{equation}

 The DFAO outputs 0 on states $q_i, 0 \le i \le m-1,$ and 1 otherwise.

 \item[(b)] Moreover, this DFAO is minimal.  
 \item[(c)-(e)] For   the remainder of the theorem statement, we need to first define sequences and recall a result.
For $m \ge 3,$   define index sequences $ (K'_i)_{0 \le i \le 2m-1}$ by 
\begin{align}\label{equ:2kernel}
K'_i = \begin{cases}
     (2^i n + 2^i-1)_{n \ge 0},				&   i \in \{0, \dotsc, m\} \\
     (2^i n + 2^i-1-2^m)_{n \ge 0}, 			&
  i \in \{m+1, \dotsc, 2m-1\}.
\end{cases}   
\end{align} 
The $(s_m[K'_i])_{i \in \{0,\dotsc,2m-1\}}$    are part of  the 2-kernel of $s_m.$ By a theorem of \cite{Eilenberg}, for
any $m,$ if $s_m$ is automatic, the distinct elements of the 2-kernel are finite and form a DFAO accepting $s_m$  whose transition function is given by \eqref{equ:deltaK}. 
\begin{alignat}{2}\label{equ:deltaK}
\delta_K(s_m[K'_i],a)=
\begin{cases}
s_m[(2^{i+1}n+2^i-1)_{n \ge 0}], 			&\qquad \text{ if } 0 \le i \le m,a=0\\
s_m[(2^{i+1}n+2^{i+1}-1)_{n \ge 0}], 		&\qquad \text{ if } 0 \le i \le m,a=1 \\
 s_m[(2^{i+1}n+2^{i}-2^m-1)_{n \ge 0}], 		&\qquad \text{ if } m+1 \le i \le 2m-1,a=0\\
s_m[ (2^{i+1}n+2^{i+1}-2^m-1)_{n \ge 0}], 	&\qquad \text{ if } m+1 \le i \le 2m-1,a=1. 
\end{cases}
\end{alignat}

\item[(c)] Let
\begin{equation}\label{equ:l}
            l = 2^{m-1}
\end{equation}
and define  
\begin{equation}\label{equ:2kernelprime}
K_i = \text{ length-$l$ prefix of $K'_i$}, \qquad 0 \le i \le 2m-1.
\end{equation}
The characteristic functions of the length-$l$ sequences $s_m[K_i], 0 \le i \le 2m-1,$
are given by \eqref{equ:kij} where $(V_k)_{k \ge 1}$ are the Vile numbers,
\seqnum{A003159},  and ${(J_k)}_{k \ge 0}$ is the Jacobsthal sequence, \seqnum{A001045}. 
For  $0 \le j \le l-1,$
 \begin{equation}\label{equ:kij}
	s[ K_{i,j}]=1 \leftrightarrow 
        \begin{cases}
        j \in \emptyset, 							&  i=0,\\ 
        j \in 2^{m-i}(V_k)_{1 \le k \le J_i}-1,  			& 1 \le i \le m, \\
        \bar{K}_{i-(m+1),j}=0, 						& m+1 \le i \le 2m-1.
\end{cases}
\end{equation} 
\item [(d)] The $s_m[K_i], 0 \le i \le 2m-1,$ are distinct.

\item [(e)] Under the correspondence, $q_i \leftrightarrow s_m[K_i], i \in \{0,\dotsc, 2m-1\},$ the automata defined by \eqref{equ:delta} and \eqref{equ:deltaK} are equivalent. More specifically,
\begin{equation}\label{equ:deltadeltaK}
\text{For $0 \le i \le 2m-1$ and $a \in \{0,1\}$,} \delta(q_i,a)=q_j \leftrightarrow
		\delta_K(q_i, a) =L \text{ and }    s_m[L] =s_m[K_j].
\end{equation} 

\end{enumerate}
\end{theorem}

Figure \ref{fig:s4} and Table \ref{tab:kernel} illustrate several parts of the theorem. 
\begin{center}
\begin{figure}[!ht]

\begin{tikzpicture}

\draw (-3,12) circle (.75cm);  %q0		-3,12

\draw (0,12) circle (.75cm); 	%q1		0,12
\draw (0,9) circle (.75cm); 	%q2		0,9
\draw (0,6) circle (.75cm); 	%q3		0,6
\draw (0,3) circle (.75cm);	%q4		0,3

\draw (3,3) circle (.75cm);	%q5		3,3
\draw (6,3) circle (.75cm);	%q6		6,3

\draw(6,6) circle (.75cm);	%q7		6,6

\draw [arrow] (-2.25,12) -- node[anchor=south] {1} (-0.75,12);		%q0->q1
\draw [arrow] (0,11.25) -- node[anchor=east] {1} (0,9.75);			%q1->q2
\draw [arrow] (0,8.25) -- node[anchor=east] {1} (0,6.75);			%q2->q3
\draw [arrow] (0,5.25) -- node[anchor=east] {1} (0,3.75);			%q3->q4
\draw [arrow] (0,3.75) -- node[anchor=west] {1} (0,5.25);			%q4->q3

\draw [arrow] (0.75,3) -- node[anchor=south] {0} (2.25,3);			%q4->q5
 \draw [arrow] (3.75,3) -- node[anchor=south] {1} (5.25,3);			%q5->q6

\draw [arrow] (6,3.75) -- node[anchor=east]  {1} (6,5.25);			%q6->q7

 \draw[->] (-3,12.75)  arc (-90:270:.35)
 node[anchor=south] {0} (-3,12.75);			%loop at q0
  \draw[->] (3,2.25)  arc (90:-270:.35)
 node[anchor=north] {0} (3,2.25);				%loop at q5

  \draw [arrow] (-.35,11.4) -- node [anchor=south]
 {0} (-3,11.25);							%q1->q0
 \draw [arrow] (-.75,9) -- node [anchor=south]
 {0} (-3,11.25);							%q2->q0
  \draw [arrow] (-.75,6) -- node [anchor=south]
 {0} (-3,11.25);							%q3->q0

  \draw [arrow] (5.25,6) -- node [anchor=south]
 {1} (0,3.75);									
   \draw [arrow] (5.35,3.5) -- node [anchor=south]
 {0} (3,3.75);
   \draw [arrow] (5.25,6) -- node [anchor=south]
 {0} (3,3.75);
  \draw [arrow] (-.75,6) -- node [anchor=south]
 {0} (-3,11.25);

\node at (-3,12) {0/0};		%label q0
\node at (0,12) {1/0};		%label q1
\node at (0,9) {2/0};		%label q2
\node at (0,6) {3/0};		%label q3
\node at (0,3) {4/1};		%label q4
\node at (3,3) {5/1};		%label q5
\node at (6,3) {6/1};		%label q6
\node at (6,6) {7/1};		%label q7

\end{tikzpicture}

\caption{The DFAO for $s_4.$}\label{fig:s4}
\end{figure}
\end{center}

\begin{center}
\begin{table}[!ht] 
\begin{center}
\begin{tabular}{||c|c|||} 
\hline \hline
Index Set, $K$ & $s_m[K]$\\ 
\hline \hline
$K_0$ & $00000000$\\
$K_1$ &$00000001$\\
$K_2$ &$00010000$\\
$K_3$ & $01000101$\\
$K_4$ &  $10111010$\\
$K_5$ &  $11111111$\\
$K_6$ & $11111110$\\
$K_7$ &  $11101111$\\ 
\hline \hline
\end{tabular}
\caption
{ These eight, distinct, length-8 prefixes of the 2-kernel sequences for $s_4,$  
illustrate Theorem \ref{the:main}(c)-(e).  }
\label{tab:kernel} 
\end{center} 
 \end{table}
\end{center}

\section{Proof of Theorem \ref{the:main}(a), (b)}
\begin{proof}  
Let $p=P_j \cdot 0 \cdot t$ with $j \ge 1,$ and $t$ a non-empty binary string. (We must separately 
consider the case $p=P_j$ but the proof is similar and hence omitted). By \eqref{equ:delta} $s_m[t] \in \{0,1\}.$
Table \ref{tab:smoutput} describes the state 
at which the DFAO defined by \eqref{equ:delta} terminates after inputting $p.$  It is then immediately clear that this DFAO
counts the number of possibly overlapping occurrences of $P_m$ in $p,$ completing the proof.

\begin{center}
\begin{table}[!ht] 
\begin{center}
\begin{tabular}{||c|c|c||} 
\hline \hline
\text{Value of $j$} 				&	$s_m[t] = 0$				& 	$s_m[t] = 1$\\ 
\hline 
 $j \in \{1
 ,\dotsc,m-2\}$			& 	$j$						&	$m+1+j$			\\
$j=m-1$						&	$m-1$					&	$m$				\\
$  j =m+e$					&	$m$						&	$m-1$			\\
$  j = m+e+1$					&	$m-1$					& 	$m$				\\
\hline 
\end{tabular}
\caption
{ Terminal state of the DFAO described by \eqref{equ:delta}  on input $p= P_j \cdot 0 \cdot t.$  }
\label{tab:smoutput} 
\end{center} 
 \end{table}
\end{center}

$\mathbf{Proof \; of \; Theorem \ref{the:main}(b).}$ The most straightforward proof uses the underlying idea in 
 the algorithm for computing the minimal state (e.g. \cite{Valmari}). More specifically, it is easy to check that no two states are compatible, that is, for any $i,j \in \{0,...,2m\}, 
i \neq j,$ $\{\delta(q_i,0), \delta(q_i,1)\} \neq \{\delta(q_j,0), \delta(q_j,1)\}.$  Alternatively, the proof of the minimality of states can  be justified by presenting inputs which differentiate states.   For example, for $1 \le i \le m-1,$ inputting  $P_{m-i}$ to state $q_i$   would result in output 1, while inputting $P_{m-i}$ to   state $q_j, 0 \le j \le i-1,$ would result in output 0; hence, state $q_i$ cannot be eliminated. Similar arguments apply to the remaining states. 
\end{proof}

\section{Some Preliminary Lemmas}

Theorem \ref{the:main}(c) provides an explicit form for the $(K_i)_{0 \le i \le 2m-1}$ which facilitates proving parts (d) and (e).  To prove part (c), we will first need three  number-theoretic lemmas.

\begin{lemma}\label{lem:2ivkm1} 
For $i \ge 0,$ 
$$
2^i (V_k)_{k \ge 1} -1 = \{n \ge 0:n_2 \text{ has a  suffix } 01^e 1^i\}.
$$
 
\end{lemma}
\begin{proof}
\begin{align*}
	(V_k)_{k \ge 1} 			&= \{n \ge 0: n_2 \text{ has suffix }   10^e\} 	 &\rightarrow \\
 	2^i (V_k)_{k \ge 1, i \ge 0}	&= \{n \ge 0 : n_2 \text{ has suffix }   10^e0^i \}  &\rightarrow \\ 
	2^i (V_k)_{k \ge 1, i \ge 0}-1 &= \{n \ge 0: n_2 \text{ has suffix } 01^e 1^i \}.    &
\end{align*}
\end{proof}

\begin{lemma}\label{lem:jkle2i} For $i \ge 0,$
$$\#(V_k)_{1 \le k \le 2^i} = J_{i+1}.$$
\end{lemma}
\begin{proof}
 For non-negative integer $n,$ $n_2$ and $1 \cdot n_2$ have the same number of trailing $0s$ (without loss of generality, we may left pad $n_2$ with zeroes to achieve a uniform length).  This implies that for $i \ge 0$ 
$$\#\{V_k: 1 \le V_k \le 2^{i}-1\} =
\# \{V_k: 2^{i}+1 \le V_k \le 2^{i+1}-1\}.$$ Trivially, 
 for $i \ge 0,$ $2^i \in (V_k)_{k \ge 1}$ iff $i \equiv 0 \pmod{2}.$

It immediately follows that for $i \ge 0,$
$$L_i = \#\{V_k:1 \le V_k \le 2^i\}=
\#\{V_k:2^i+1 \le V_k \le 2^{i+1}\} +(-1)^i,$$
and therefore $L_i$ satisfies the recursive relationship 
$$ L_{i+1} = 2L_i - (-1)^i, \qquad i \ge 0.$$  But then $(L_i)_{i\ge 1}$ must also satisfy the Jacobsthal recursion since for $i \ge 0,$
$$L_{i+1} + 2L_{i} = 2L_{i+1} +(-1)^i = L_{i+2}.$$
The inductive proof is completed by confirming the base case, $\# L_0 = J_1.$ 
\end{proof}

\begin{lemma}\label{lem:jkpowerof2} $$J_{m-1} + J_m  = 2^{m-1},m \ge 1.$$
\end{lemma}
\begin{proof} The  Jacobsthal recursion implies $J_{m-1}+J_m = 2(J_{m-1} + J_{m-2})$ showing that the expression
$J_{m-1}+J_m$ satisfies the same recursion satisfied by the sequence  $(2^i)_{i \ge 0}.$ The inductive proof is completed by confirming the base case when $m=1$  or $m=2.$
\end{proof}

\section{Proof of Theorem \ref{the:main}(c)}

We must prove each of the three cases to the right of the braces listed in \eqref{equ:kij}.

\textbf{The case $i=0.$} Since, by \eqref{equ:l}, $(P_m)_v = 2^m-1 > l - 1 = 2^{m-1}-1$ it immediately follows that   
$$K_0 = (n)_{0 \le n \le l-1} = 0^l.$$ 

\textbf{The case $1 \le i \le m.$} Suppose for some $j,$ $1 \le n 
\le l,$ that	 $j=2^{m-i}V_k,$ for some $k$.  
Then by \eqref{equ:2kernel} and \eqref{equ:2kernelprime}
\begin{align*}
&\; K_{i,j-1}=1 && \leftrightarrow \;  \\
&\;  2^i(j-1)+2^i-1 \text{ has an odd number of occurrences of $P_m$}
& \leftrightarrow\\
&\; (2^i(j-1)+2^i-1)_2 \text{ has a suffix $0 1^e  1^m$} &
\leftrightarrow \\
&\; 2^i(j-1)+2^i-1 = 2^m V_k -1 \text{ for some $k$},  &
\leftrightarrow\\
&\; j = 2^{m-i} V_k \text{ for some $k$},
\end{align*}
the next to last equivalence following from Lemma \ref{lem:2ivkm1}.

\textbf{The case $m+1  \le i \le 2m-1.$} For each $i,$ $0 \le i \le m-2,$ 
$$K_i = \{(2^i n+2^i-1)_2: n \in \{0,\dotsc, l-1\}\}  =\{ n_2 \cdot 1^i: n \in \{0,\dotsc, l-1\}\}$$ while
$$K_{m+1+i} = \{2^{i+m+1} n + 2^{i+m+1}-2^m-1)_2: n \in \{0,\dotsc, l-1\}\} =
\{n_2 \cdot 1^i \cdot 0 \cdot 1^m: n \in \{0,\dotsc, l-1\}\}.$$  

It immediately follows that $K_{m+1+i}$ has one extra occurrence of $P_m,$ and therefore, the number of occurrences of $P_m$ in $K_i$ has opposite parity to the number of occurrences in $K_{m+1+i}.$

\section{Proof of Theorem \ref{the:main}(d)}

By inspection, the result is true for $m=3.$ We therefore assume for the rest of the proof that $m \ge 4.$
The proof consists of a collection of cases, according to the index of $K_i, 0\le i \le 2m-1.$

\begin{itemize}
\item \textbf{Distinctness of $K_i, i=0,\dotsc, m.$} By Theorem \ref{the:main}(c), $K_0$ has no 1s, while for $i \ge 1,$ the first 1 in $K_i$ occurs at position $2^{m-i}-1.$  
\item \textbf{Distinctness of $K_i, i =m+1, \dotsc, 2m-1.$} By Theorem \ref{the:main}(c), $K_i = \bar{K}_{i-(m+1)}.$  Therefore, the distinctness of the 
$K_i$ in the indicated range follows from the just-proved distinctness of the $K_{i-(m+1)}$ in that range.
\item \textbf{Distinctness of $K_i, i=0,\dotsc,m$ and $K_i, i=m+1,\dotsc,i=2m-1.$} This case has 3 subcases.
\begin{itemize} 
\item \textbf{The case of $i=0.$} By Theorem \ref{the:main}(c), $K_0$ has no 1s while all other $K_i, 1 \le i \le 2m-1$ have 1s implying $K_0$ is distinct from them. 
\item \textbf{The case $1 \le i \le m-1.$} By Theorem \ref{the:main}(c), the first 1 in $K_i,$ occurs at position $2^{m-i}-1 \neq 0,$
while the first one in $K_i, m+1 \le i \le 2m-1,$ occurs at position 0 implying they are distinct. 
\item \textbf{The case $i=m.$} We suppose to the contrary that 
 $K_m=K_{m+1+j}$  for some $j \in \{0, \dotsc, m-2\}\}$ and derive a contradiction. 
By Theorem \ref{the:main}(c), $K_j = \bar{K}_{m+1+j},$ implying
$K_m + K_j = 1^m,$ (the ``+" indicating component-wide addition of vectors and $1^m$ is regarded as the vector of $m \; ones$). 
By Theorem \ref{the:main}(c),  the number
of ones in $K_k, 1 \le k \le m$ equals $J_k.$ 
This implies  $J_m + J_j = 2^{m-1}=J_m +J_{m-1},$ the last equality justified by 
Lemma  \ref{lem:jkpowerof2}. But then for some j, $0 \le j \le m-2,$
$J_j = J_{m-1}$ contradicting the strict monotonicity of $(J_m)_{m \ge 3}.$
This completes the proof.
\end{itemize}
\end{itemize}  
These cases taken together show that the $K_i, 0 \le i \le 2m-1$ are all distinct, completing the proof
of Theorem \ref{the:main}(d).

 \section{Proof of Theorem \ref{the:main}(e).}

We must show \eqref{equ:deltadeltaK}.  This, in turn, requires checking seven cases 
which are presented in Tables \ref{tab:e1}-\ref{tab:e3}.   
The final step of checking that the parity of the  number of occurrences of $P_m$
in $L$ and $K_j$ are equal is straightforward; for example, for Column E, this follows from the equal parity of
$m-1$ and $m+1,$ while for Column C it follows from the fact that $i-(m+1) <m.$

The proof of the Main Theorem is complete.

\begin{center}
\begin{table}[!ht] 
\begin{center}
\begin{tabular}{||c|c|c|c||} 
\hline \hline
ID & A & B & C \\ 
\hline 
Range $\;i$ & 			$0 \le i \le m-1$ & 			$i=m$ & 					$m+1 \le i \le 2m-1$ \\
\hline 
$\delta(q_i,0)=q_j$ & 	$q_j=q_{0}$ &	 		$q_j=q_{m+1}$ & 			$q_j=q_{m+1}$ \\
 $(K_j)_2				$&$n_2$&				$n_2 \cdot 0 \cdot 1^m$&	$n_2 \cdot 0 \cdot 1^m$\\		
\hline 
$K_i$ & 				$(2^i n + 2^i-1)$ & 		$2^m n + 2^m-1$ & 		$2^{i}n+2^{i}-2^m -1 $ \\
$(K_i)_2$ & 			$n_2 \cdot 1^i$ & 			$n_2 \cdot 1^m$ & 			$n_2 \cdot 1^{i-(m+1)} \cdot 0 \cdot 1^m$\\
$\delta_K(s_m[K_i],0)=L$ &	$L=2^{i+1}n+2^i-1$ &		$L=2^{m+1}n+2^m -1$ &		$L=2^{i+1}n + 2^{i}-2^m-1 $ \\
$L_2$ 			&$n_2 \cdot 0 \cdot 1^i$ &	$n_2 \cdot 0  \cdot 1^{m}$&$n_2 \cdot 0 \cdot 1^{i-(m+1)} \cdot 0 \cdot 1^m$ \\ 
\hline \hline
\end{tabular}
\caption
{Three of the seven cases proving Theorem \ref{the:main}(e). Notation follows \eqref{equ:2kernel}, \eqref{equ:2kernelprime}, and \eqref{equ:deltadeltaK} . Table
\ref{tab:e1} is continued in Tables \ref{tab:e2}- \ref{tab:e3}.}
\label{tab:e1} 
\end{center} 
 \end{table}
\end{center}

\begin{center}
\begin{table}[!ht] 
\begin{center}
\begin{tabular}{||c|c|c||} 
\hline \hline
ID & D & E \\
\hline 
Range $\;i$ & 			$0 \le i \le m-1$ & 			$i=m$  		 \\
\hline
$\delta(q_i,1)=q_j$ & 	$q_j=q_{i+1}$ & 			$q_j=q_{m-1}$  \\
 $(K_j)_2$&			$n_2 \cdot 1^{i+1}$ & 		$n_2 \cdot 1^{m-1}$  \\
\hline
$K_i$ & 				$2^i n + 2^i-1$ & 			$2^m n + 2^m-1$  \\
$(K_i)_2$ & 			$n_2 \cdot 1^i$ & 			$n_2 \cdot 1^m$  		 \\
$\delta(K_i,1)=L$ &	$L=2^{i+1}n+2^{i+1}-1$ &	$L=2^{m+1}(n+1)-1$ \\		
$L_2$  &				$n_2  \cdot 1^{i+1}$ &		$n_2 \cdot 1^{m+1}$ \\ 
\hline \hline
\end{tabular}
\caption
{The next  two of the 7 cases proving Theorem \ref{the:main}(e), started 
in Table \ref{tab:e1} and completed in Table \ref{tab:e3}.  Notation follows \eqref{equ:2kernel}, \eqref{equ:2kernelprime}, and \eqref{equ:deltadeltaK} .}\label{tab:e2} 
\end{center} 
 \end{table}
\end{center}

\begin{center}
\begin{table}[!ht] 
\begin{center}
\begin{tabular}{||c|c|c||} 
\hline \hline
ID &F & G\\
\hline 
Range $\;i$ & 			$m+1 \le i \le 2m-2$					&$i=2m-1$ \\
\hline
$\delta(q_i,1)=q_j$ & 	$q_j=q_{i+1}$ 						&$q_j=q_m$\\
 $(K_j)_2$&			$n_2 \cdot 1^{i-m} \cdot 0 \cdot 1^m $		&$n_2 \cdot 1^m$ \\\hline
$K_i$ & 				$2^{i}(n+1) -2^m -1 $					 &$2^{2m-1}(n+1) -2^m-1$\\
$(K_i)_2$ & 			$n_2 \cdot 1^{i-m-1} \cdot 0 \cdot 1^m$	&$n_2 \cdot 1^{m-2} \cdot 0 \cdot 1^m$	\\
$\delta(K_i,1)=L$ &	$L=2^{i+1}(n+1)-2^m-1 $	 			&$L=2^{2m} n+ 2^{2m}-2^m - 1$\\		
$L_2$  		&		$n_2 \cdot 1^{i-m} \cdot 0 \cdot 1^m$ 		&$n_2 \cdot 1^{m-1} \cdot 0 \cdot 1^m$ \\ 
\hline \hline
\end{tabular}
\caption
{The last 2 of the 7 cases proving Theorem \ref{the:main}(e), continued from 
Tables \ref{tab:e1} and \ref{tab:e2}. Notation follows \eqref{equ:2kernel}, \eqref{equ:2kernelprime}, and \eqref{equ:deltadeltaK} .}\label{tab:e3} 
\end{center} 
 \end{table}
\end{center}

\section
{Attributes and Properties of 
\texorpdfstring{$s_m$}{Sm}}
 
\label{sec:proofmethods}
One contribution of this paper is the exploitation of the correspondence between integers under addition and multiplication and strings under concatenation. To show the usefulness of this method we generalize known results about maximal runs, palindromes, squares, and unbordered factors in the Thue-Morse and Rudin-Shapiro sequences  to all $s_m, m \ge 3.$   Although, as mentioned in the introduction, the double sequence $\{(s_{m,n})_{m \ge 1, n\ge 0}\}$ is not automatic, we can still prove many general results. Many theorems share the
following  common \textit{3-step}
proof method.

\begin{enumerate}
\item [(I)] Reformulate the target result to be proven into a collection of assertions about the values of $s_{m,n}, n \in I_m$ with $I_m$ some indexing set.
\item[(II)] For each  $n \in I_m,$ calculate $n_2.$
\item[(III)] Using the binary representations of Step (II), prove the target result of Step (I)  by counting (possibly overlapping) occurrences of $P_m$ in each $n_2, n \in I_m.$
\end{enumerate}

Because of the similarity of many proofs, in the sequel, we prove: (i) one component of one  theorem, Theorem \ref{the:runs}(e), the proof of 
many of the other components of that and other theorems being similar and hence omitted; and (ii) assertions whose proof method involves non-routine
applications of the 3-step method just presented.

A separate technique is useful for formulating theorems. For example,   to prove a theorem about lengths of
all maximal runs or palindromes,   it is useful to use a \textit{prove more to establish less} method whereby instead of just describing the maximal lengths we are interested in, we describe triples $(k,b,x)$ such that $s_m[b..b+k-1]=x$ and $x$ is a maximal run or palindrome.  Because
patterns for $k,b,x$ are often present, proofs are easier to accomplish. Towards this end of describing patterns in triples,
$k,b,x,$ exploratory pattern analysis for results are done with both Walnut and Mathematica 13.3 which has very powerful pattern matching functions. The following two lines of Mathematica code define $P_m$ and the infinite 
double array ${(s_{m,n})}_{m \ge 1, n \ge 0}.$

{\tolerance 1000
\begin{verbatim} 
P[1] = "1"; P[m_] := P[m] = P[m - 1] <> "1"; 
s[m_, n_] := Mod[StringCount[IntegerString[n, 2], P[m], Overlaps -> True], 2];
\end{verbatim}} 

The sequences can then be decimated to produce and explore patterns in  the length $l$-prefixes of the 2-kernel sequences.
The respective use of these software will be indicated in discussions of each theorem.

\section{Run Positions and Lengths}

We  study maximal runs of $0s$ in $s_m.$ If $p=1 \cdot 0^k \cdot 1=s_m[b .. b+k+1]$  then we say that $0^k$ is a maximal run
of zeroes  of length $k$ starting at position $b+1$ (maximal runs starting at 0 are also allowed and defined analogously). Using Walnut it is straightforward to  show that the longest run of zeroes in $s_2,$ the Rudin-Shapiro sequence, is 4 \cite[Section 8.1.9]{Shallit}. The generalization of this to $s_m$ is that the longest run of zeroes in $s_m$ is $2^m.$ For any particular $m \ge 3,$ we can,  using Walnut, discover all maximal run lengths of $s_m,$ and, again using Walnut, prove that these discovered maximal run lengths are the only ones.  Theorem \ref{the:runs} proves a general result, valid for all
$m \ge 3,$ describing the five types of maximal run lengths. To prove the result  we use the
\textit{prove more to establish less} method.  First, Walnut was used to discover, for individual small $m,$ 
all maximal run lengths, and then, for each discovered maximal run
length, Walnut was used to find a starting position for that maximal run length.  Patterns for the beginning positions and lengths emerged, summarized in Theorem \ref{the:runs}.

\begin{theorem}\label{the:runs}  Fix $m \ge 3,$  and let $IsRun_m(k,b),$  
be the first order statement that $s_m$ has a maximum run of zeroes of length $k$ beginning  at position $b.$  $IsRun_m(k,b)$ is true for all pairs in the following table. 
 \begin{center}
\begin{table}[!ht] 
\begin{center}
\begin{tabular}{||c|c|c|c|c||} 
\hline \hline
Theorem \; ID & $k$ & $b$ & Range\;of\;$i$\\
\hline 
(a) 	&	$2^m$			&	$2^{m+3}-1$					&	\; 			\\
(b)	&	$2^{m-1}$ 		& $(2^{m-1}-1)(2^{m})$			&	\;			\\
(c) 	&	$2^i$			&	$2^{i+1}(2^{m+1}-1)$			&	$0 \le i \le m-2$\\
\hline 
(d)	&	$2^{m-1}+1$		&	$(2^m-1)(2^{m+1} + 1)$		&				\\
\hline 
(e) 	& 	$2^m-2^i$		&	$(2^i-1)2^m$					&  $0 \le  i \le m-2.$ \\
\hline \hline	
\end{tabular}
\caption
{ A list of  cases for which $IsRun_m(k,b)$ is true. Conjecture \ref{con:runs} asserts that these are the only  maximal run lengths.}
\label{tab:runs} 
\end{center} 
 \end{table}
\end{center} 
\end{theorem} 

For the general case of arbitrary $m \ge 3,$  we can prove the assertions of the table, but are unable to prove that the lengths enumerated in the table are the only maximal run lengths. We therefore leave this as an open conjecture, noting, that the conjecture is strongly supported by the Walnut-proven fact that it is true for many initial $m.$  The following is sample Walnut code for $m=4$ which can be adapted to any other specific $m \ge 3.$

\begin{verbatim}
def isrun "k>0 & Eb ((b=0 | S4[b-1]=1) & Aj ((j<k) => (S4[b+j]=0)) & S4[b+k]=1)":
eval theorem "Ak $isrun(k) <=>k=1 | k=2| k=4| k=8 |k=9 | k=12 | k=14 |k=15 | k=16":
\end{verbatim}

\begin{conjecture}\label{con:runs} For each $m \ge 3,$ the  lengths presented in Theorem \ref{the:runs} are   the only maximal run lengths.
\end{conjecture}

We have left to prove Theorem \ref{the:runs}.
\begin{proof}
We suffice, using the 3-step method indicated in Section \ref{sec:proofmethods}, with proving one component of the theorem, in this case, Theorem \ref{the:runs}(e), the proof of the other cases being similar
and hence omitted. We therefore suppose given integers $m,i, 1  \le i \le m-2.$ 

\textbf{Step I:} Let $b=(2^i-1)2^m$ and let $k= 2^m-2^i.$ To prove assertion (e) it suffices to prove the following assertions:
\begin{itemize}
\item $s_m[b-1]=1$
\item $s_m[b]=0$
\item $s_m[b+j] = 0, 1 \le j <   k$    
\item $s_m[b+k] = 1.$
\end{itemize}

\textbf{Step II:} For $1 \le i \le m-2,$ we obtain the following binary representations.
\begin{itemize}
\item $(b-1)_2= 1^{i-1} \cdot 0 \cdot 1^m.$
\item $b_2 = 1^i \cdot 0^m.$
\item $(b+j)_2 = 1^i \cdot j', \text{ where } |j'| =|j_2| \text{ and } j'_v = j, 1 \le j \le   k-1$ ($j'$ equals $j_2$ left-padded with zeroes)    
\item $(b+k)_2 = 1^m \cdot 0^i.$
\end{itemize}

\textbf{Step III:} By inspection, the assertions in Step I immediately follow from the binary representations in Step II.
\end{proof}

\section{Squares}

Recall that a factor of the form $xx$ of a word is called a square and $|x|$ is the order of the square. It is well known that the orders of the squares in the Thue-Morse sequence are $(2^n)_{n\ge 0} \cup (3 \times 2^n)_{n \ge 0}.$ This section generalizes this result for $(s_m)_{m \ge 3}.$
There are three cases to consider corresponding to the next three lemmas. The patterns discussed in 
the last two of these three lemmas are more apparent
in binary representations of numbers, and hence, these lemmas are formulated in terms of binary strings. The proof methods, although based on the 3-step proof method, are slightly different
for each lemma.  The first result is straightforward.  

\begin{lemma}[Trivial] For $m \ge 3,$ there are squares of order $1, \dotsc, 2^{m-1}-1$ in $s_m.$ \end{lemma}
\begin{proof} Be definition,  the first 1 in $s_m$ occurs  at $P_m.$ Since $(P_m)_2 =2^m-1,$   
$s_m[0 .. 2^m-1] = 0^{2^m-1} \cdot 1,$ implying that
$0^i \cdot 0^i, 1 \le i \le 2^{m-1}-1$ are square factors of $s_m.$
\end{proof}.

\begin{lemma}[Near2tom] For $m \ge 3$ and $0 \le i \le m-1,$ 
Table \ref{tab:near2tom} 
presents triples of binary strings $(k,b,x)$ 
for which the first order statement 
\begin{equation}\label{equ:ordernear2tom}
s_m[b .. b+k-1] = x = s_m[b+k.. b+2k-1]
\end{equation}
is true.

\begin{center}
\begin{table}[!ht] 
\begin{center}
\begin{tabular}{||c|c|c|||} 
\hline \hline
$k$							&$b$									&$x$		 	 			\\
\hline
$1 \cdot 0^{{m-1}}  \cdot 1$		 &$1\cdot  0^{m}$				&$  0^{2^m-2} \cdot 1\cdot 0\cdot 0$		\\
$1^i \cdot  0^{m-1-i} \cdot  1$ 	&$1^{m-1-i} \cdot   0^{i+1} \cdot  1^{m-i}$ &$	0^{2^{m-i}(2^i-1)} \cdot  1$		\\
 
\hline \hline	
\end{tabular}
\caption
{For $m\ge 3,$ and $0 \le i < m$ each row of binary string triples $(k,b,x)$ satisfies \eqref{equ:ordernear2tom}. }
\label{tab:near2tom} 
\end{center} 
 \end{table}
\end{center} 
\end{lemma} 
\begin{proof} The proof of the first line of the table follows 
the 3-step method presented in Section \ref{sec:proofmethods} and hence is omitted.
The proof of the second line, while still following this 3-step method, 
needs some extra attention since the number of individual cases grows unbounded as $m$ goes to infinity.
For the proof, we fix $m \ge 3,$ and further fix some $i$ with $0 \le i \le m-1.$ We must first show
that 
\begin{equation}\label{equ:toprovenear2tom}
\begin{cases}
\begin{aligned}
s_m[b+i] & =  0, \qquad \text{ and } s_m[b+k+i] = 0, &	\qquad \text{ for } i  < 1^i \cdot  0^{m-i},  \\
s_m[b+i]& = 1,   \qquad \text{ and } s_m[b+k+i] = 1, &	\qquad \text{ for } i = 1^i \cdot  0^{m-i}.
\end{aligned}
\end{cases}
\end{equation}

By the lemma hypotheses,
$b=1^{m-1-i} \cdot  0^{i+1} \cdot      1^{m-i},$
and hence while $b$ has no occurrences of $P_m,$ 
$b+k-1 = 1^{m-1-i}  \cdot  0^{i+1} \cdot     1^{m-i}+1^i  \cdot  0^{m-i} = 1^{m-1-i} \cdot  0  \cdot  1^m$
has one occurrence of $P_m.$ Furthermore, for $i < k-1,$
$b+i = b+k-1 -(k-1-i) = 1^{m-1-i} \cdot  0  \cdot  1^m -(k-i)$ has no occurrences of $P_m$ since $k-i$ is positive
with less than $m$ binary digits. As required,   $x=  (1^i \cdot 0^{m-i})_v =2^{m-i} 2^i-1.$  
This completes the proof.
\end{proof}
 
The final lemma deals with powers of two. By the Trivial Lemma, if $k \in (2^i)_{0 \le i \le m-2},$ there is a square
of order $k$ in $s_m.$ So we only need to deal with the cases $i \ge m-1.$ We first mention two reasonable 
but failed attempts to prove
that there are squares of all orders $\{2^i: i \ge m-1\}$ in $s_m.$ We then present a 3rd approach which works.

Approach 1, uses the \textit{prove more to establish less} method. Table \ref{tab:failure}  presents the 
results of this attempt.   Although, for several individual $i \ge m$ a pattern emerges for $x$ over all $m,$  no uniform 
pattern emerges over $x$ for both all $m$ and all $i \ge m;$  therefore, approach \;1 fails.  

\begin{center}
\begin{table}[!ht] 
\begin{center}
\begin{tabular}{||c|c|||} 
\hline \hline
$i$			 	&$x$\\
\hline
 $m$		&$0 \cdot 1  \cdot 0^{2^m-2}$ \\
$m+1$		&$0^{2} \cdot  1 \cdot  0^{2^m} \cdot  1 \cdot  0^{2^m-4}$ \\
$m+2$		&$0^4 \cdot  1^2 \cdot  0 \cdot  1 \cdot  0^{2^m-1}\cdot  1 \cdot  0^{2^m-2} 
			\cdot 1 \cdot  0^{2^m} \cdot 1 \cdot  0^{2^m-8}$ \\	
 \hline \hline	
\end{tabular}
\caption
{For $m \ge 3, i \ge m, k=1 \cdot  0^i$, and $b=1^{m+1} \cdot  0^{i+1-m},$ 
computations for several initial $m$ show that $s_m[b..b+k-1]=x$ and $xx$ is a square factor of $s_m.$
However, no apparent patterns for both all $x$ and all $i$ emerge. }
\label{tab:failure} 
\end{center} 
 \end{table}
\end{center} 

Approach 2, seeks other strings $y \neq x$ with $x$ as in Table \ref{tab:failure}, 
such that there is a square $yy$ of order $k$ in $s_m.$ The hope in this approach is
that perhaps there is a pattern over a cleverly selected sequence of squares. However, it turns out, that for the small $k,m$ reviewed, 
there are exactly two distinct squares, $xx$ and $yy$ of order $k$ in $s_m, m \ge 3,$  and moreover $x=\bar{y}.$  
This pretty result is stated (in one direction) as part of  Lemma  \ref{lem:Power2}; at the end of the section, it is proven by Walnut for small $m$ 
accompanied by a conjecture for the general case. However, as a consequence, Approach 2, seeking a cleverly selected sequence of  squares, also does not work. 
However, a 3rd approach  does work.   
To prove $xx$ is a square with $s_m[b .. b+k-1] =x= s_m[b+k ..b+2k-1]$ for some $b,k$, 
it is not necessary to identify patterns in the squares, but rather, it suffices to show 
\begin{equation}\label{equ:toprovesquares} s_m[b+j] = s_m[b+k+j], \qquad 0 \le j \le k-1.\end{equation} 
This in turn can be established using the 
3-step proof method described in Section \ref{sec:proofmethods}.

After stating and proving this lemma, we will show using Walnut, that for
small values of $m,$ 
the three lemmas completely describe all square orders of $s_m$ leading to a conjecture that this is true for all $m.$

\begin{lemma}[Power2]\label{lem:Power2} Fix $m \ge 3$ and $i \ge m-1,$ and define
$$
k=1 \cdot 0^i, b1= 1^{m+1} \cdot  0^{i+1-m},  b2=1^{m-1} \cdot  0 \cdot  1^{m+2} \cdot  0^{i+1-m}, 
x = s_m[b1 .. b1+k-1], \text{ and } y= s_m[b2.. b2+k-1].$$
Then $xx$ and $yy$ are square factors of $s_m$ and $x=\bar{y}.$ 
\end{lemma}

 \begin{proof} The proofs that $x=\bar{y}$ and $xx$ is a square factor of $s_m$ are similar to the proof
that $yy$ is a square factor of $s_m$ and hence omitted. To prove \eqref{equ:toprovesquares},
we consider two cases.  For convenience we let $b=b2.$

 \textbf{Case 1,  $0 \le j  \le u=1^{i+1-m}.$} To show that $s_m[b+j] =s_m[b+k+j],$ define a binary string $j'$ such that
$j' = j_2$ and $|j'| = i+1-m$ (if necessary by left padding
with 0s). Then
$$
	b+j = 	1^{m-1} \cdot 0 \cdot 1^{m+2} \cdot j'		\qquad b+k+j = 1^m \cdot 0^3 \cdot 1^{m-1} \cdot  j'.		
 $$
The result follows because the corresponding exponent pairs 
(of 1) occurring in $b+j$ and $b+k+j,$
$(m-1,m)$ and $(m+2,m-1),$  have opposite parity implying that the count of (possibly overlapping) occurrences
 of $P_m$ in both of them have the same parity.
 
\textbf{Case 2, $0 \le j \le k-1 -(u+1) = 1^{m-2} \cdot 0 \cdot 1^{i+1-m}.$} Defining $j'$ with $j'_2 = j_2$ and $|j'| =i,$ we have 
\begin{equation}\label{equ:case4}
		b+(u+1) +j = 1^{m} \cdot 0^3  \cdot j' , \qquad  
		b+(u+1)+k+j =   1^m \cdot 0^2 \cdot 1 \cdot   j'.
\end{equation}
The result follows because (i) when $j$ is at its maximum,  $j=1^{m-2} \cdot 0  \cdot 1^{i+1-m},$
then, 
$b+(u+1) +j = 1^m  \cdot  0^3  \cdot 1^{m-2}  \cdot 0  \cdot 1^{i-m+1}$ and
$b+k+(u+1) +j = 1^m  \cdot 0^2  \cdot 1^{m-1}  \cdot 0  \cdot 1^{i-m+1};$  
corresponding non-equal exponents $(m-2,m-1)$ are bounded above
by $m$ implying they do not contribute to the count of (possibly overlapping) $P_m;$ 
consequently, (ii) if $j$ is less then its maximum, then the exponents $(m-2,m-1)$ decrease and hence, 
still do not contribute to the count of (possibly overlapping) $P_m.$
\end{proof}
   
We next explore whether the three lemmas exhaust all possibilities of orders of squares in $s_m,$ 
For $m=8,$ the following Walnut code returns
true showing that the $k$ covered by the Trivial, Near2tom, and Power2 cases are the only orders of square factors
appearing in $s_8$ which in the Walnut code is called S8.
 Similar Walnut code can be used to show this for $s_m, m=3,\dotsc, 7$ and other small $m.$

\begin{verbatim}
def trivial "l>=1 & l<=127";
def near2tom "l=129 | l=193 | l=225  | l=241  | l=249  | l=253  | l=255   | l=257";
reg power2 msd_2 "10*"; 
def square "Ai,p1,p2 (l>=1 & s>0 & p1=s+i & p2=p1+l & i<l)=> (S8[p1]=S8[p2])";
eval conjecture "Al $trivial(l)|$near2tom(l)|$power2(l) <=> 
Es s>=1 & l>=1 & $square(l,s)";
\end{verbatim}

This leads to the following conjecture.

\begin{conjecture} For $m \ge 3$ the three lemmas list all $k \ge 1$ for which $s_m$ has square factors of order $k.$ \end{conjecture}

We next return to the Power2 lemma assertion that  for $i \ge m,$ the squares of   order $2^i$ come in pairs which are
binary reversals of each other. The following Walnut code returns true proving that in $s_3,$
which in the Walnut code is called S3,
  every square factor of order
$2^i, i \ge 3$ is equal as a factor to either $xx$ or $yy,$ where
$x$ and $y$ are as defined in the Power2 lemma. Similar Walnut
code would provide proofs for other small $m.$  

\begin{verbatim}
def sq "Al,p1,p2 i<l & p1=s+i & p2=p1+l => S3[p1]=S3[p2]":  
#[s..s+l-1]=x and xx is square in S3 
def eq "Ai,s1,s2 i<l & p1=s1+i & p2=s2+i => S3[p1]=S3[p2]": 
#S3[s1..s1+l-1]= S3[s2..s2+l-1]  
def br "Ai,p1,p2 i<l & p1=s1+i & p2=s2+i=> (S3[p1]=@1 <=> S3[p2]=@0)": 
# S3[s1..s1+l-1] ,S3[s2..s2+l-1] are binary reversals of each other
reg power2 msd_2 "10*"; #The argument of power2 is a binary power
eval test "Al,s,t 
( l>=8 & $power2(l) & $sq(l,s) & $sq(l,t) & ~$eq(l,s,t))=>$br(l,s,t)": 
#The 2 distinct squares of order l given in the Power2 lemma are binary reversals
eval test "As,t,u,l 
($sq(l,s) & $sq(l,t) & $sq(l,u) & l>=8 & $power2(l)) => 
($eq(l,s,t) | $eq(l,s,u) | $eq(l,u,t))":
#There are no 3 distinct squares of order l when l is a power of 2
\end{verbatim}

These results lead to the following conjecture.
\begin{conjecture} For $m \ge 3, i \ge m-1,$ 
the Power2 Lemma lists the only two distinct square factors of order $2^i$. Moreover these two distinct factors are binary reversals of each other.\end{conjecture}   
 
\section{Proof Method Evaluated by the Arithmetic Hierarchy}\label{sec:hierarchy}

The previous two sections showed that the proof method presented in Section \ref{sec:proofmethods}) gives almost mechanical, short, compact proofs for a variety of assertions. Contrastively, the next two sections will show statements on which the proof method helps little  or not at all. 
 
A little reflection shows that the sections differ in the complexity of their respective assertions as defined by the arithmetic hierarchy (which classifies first-order logical statements by the number of alternations of quantifiers in their normal form). The \textit{prove more to establish} less method when applied in the last two sections dealt with $\Sigma_2$ statements: for example, $k$ is a maximal run length means there is some beginning position, $b,$ such that for all indices $k',$ with $1 \le k' <k,$ $s_m[b+k']=0, s_m[b]=1=s_m[b+k].$  A similar analysis applies to the assertion that $k$ is the order of a square factor of $s_m.$  The proof method, by making explicit the patterns in starting positions, reduces the
$\Sigma_2$ assertion to a $\Pi_1$ assertion, facilitating proofs. 

Contrastively, assertions about unbordered factors of a fixed length, say $k,$ meaning there is a beginning position $b,$ such that for all possible positive boundary lengths $b'<k,$ there is some index breaking the equality of a length $b'$ prefix and suffix,  are $\Sigma_3$ statements.
Removing one quantifier from a $\Sigma_3$ statement does not help that much in finding proofs as we shown in detail in coming sections. 

Palindromes offer possibilities of both $\Sigma_2$ and $\Sigma_3$ statements, the former amenable to the proof methods while the latter are not. Palindromicity of a specified length $k,$ means there exists a beginning position, $b,$ such that for all  $k' <k, s_m[b+k']=
s_m[b+k-1-k']$ (with some additional technicalities if $b=0$ which case we ignore for this discussion). 

Hence, palindromicity is a $\Sigma_2$ statement. Local maximal palindromicity simply adds the requirement $s_m[b-1] \ne s_m[b+k],$ and hence is still $\Sigma_2.$ Contrastively, global maximal palindromicity adds the requirement that for all other beginning positions $b',$ either their exists a $k'<k$ breaking the palindromicity starting at $b'$ or
$s_m[b'-1] \ne s_m[b'+k]$ and hence is $\Sigma_3.$  Proofs of assertions connected with global maximal palindromicity seem difficult to obtain using our proof methods. 

While these remarks are observational, they are consistent other studies of first-order statement complexity as measured by the arithmetic hierarchy \cite{Complexity1, Complexity2}.

\section{Palindromes}
There are two distinct definitions of maximal palindromes, \textit{locally maximal palindromes} which for a given sequence, $t,$ 
asserts that if for some $b \ge k,k\ge 1,$  $t[b .. b+k-1]$ is a palindrome then $t[b-1 .. b+k]$ is not a 
palindrome, and 
\textit{globally maximal palindromes}, which for a given sequence, $t,$ asserts that if $y=t[b .. b+k-1]$ is a palindrome then 
$aya$ is not a factor of $t$ for any letter $a.$ (Separate but similar definitions must be given for the case $b=0.$) For the Thue-Morse sequence,  the maximal palindrome lengths are $(3 \times  4^n)_{n \ge 0}$. For the Rudin Shapiro sequence, the maximal palindrome lengths are
$7, 10, 14,$ (see \cite[Section 8.6.4]{Shallit} which formulates the definition of maximal local and global palindromes differently).   

We shall study the $n$-th \textit{largest maximal palindrome lengths}, for $n=1,2,3,4,5.$ For a given sequence, a length $k$ is an $n$-th largest maximal palindrome length, if  $k$
 is a global maximal palindrome length and there are only $n-1$ distinct lengths, $k'>k$ for which there are maximal global palindromes of length $k'.$ 

It is easy with Walnut to completely discover and prove the $n$-th largest maximal global palindrome lengths for any particular $s_m.$ To prove a general theorem, we use the \textit{prove more to establish less} method. 

In exploring patterns, Walnut was used to discover the maximal palindrome lengths and appropriate starting positions, while Mathematica was used to discover the underlying palindrome patterns. For each $n \in \{1,2,3,4,5\}$ patterns emerged. It was hoped that these 5 cases of $n$ would suffice to establish a pattern for all $n;$ however, the results for each $n$ appear quite different and distinct. They are summarized in the following theorem and conjecture.

\begin{theorem}\label{the:palindromes} Fix $m \ge 3.$ Tables \ref{tab:palindromesab}-\ref{tab:palindromese}
 present quadruples $(n,k,b,x)$ such that $x=s_m[b..b+k-1]$ is a locally maximal palindrome. \end{theorem}
\begin{conjecture}\label{con:palindromes} The palindromes listed in Tables \ref{tab:palindromesab}-\ref{tab:palindromese} are the 
$n$-th largest globally maximal palindromes for
$n=1,2,3,4,5.$ \end{conjecture}
\begin{remark}  As already indicated, for $3 \le m \le 8,$ straightforward Walnut code proves Conjecture \ref{con:palindromes}. \end{remark}

\begin{center}
\begin{table}[!ht] 
\begin{center}
\begin{tabular}{||c|c|c||} 
\hline \hline
Value of $n$ 		& $n=1$	 						& $n=2$ 	\\
\hline 									 	
$k$				&$2^{m+1}+2^m+6$ 				&$2^{m+1}+2^m-6$ 		  \\
$b$				&$2^{m+3}-4$					&$2^m+2$							 \\	
$x=s_m[b .. b+k-1].$   &
$0^2 \cdot 1 \cdot 0^{2^m} \cdot 1 \cdot 0^{2^m - 2} \cdot 1 \cdot 0^{2^m} \cdot 1 \cdot 0^2$ &
$0^{2^m-4} \cdot 1 \cdot 0^{2^m} \cdot 1 \cdot 0^{2^m-4}$ \\
\hline \hline	
\end{tabular}
\caption
{The $n$-th largest global palindrome length, $k,$ for $s_m,$ $3 \le m \le 8,$ for  $n=1,2,$  with one possible starting position, $b,$ and the underlying palindrome pattern, $x.$ It is conjectured that the table holds for all $m.$}
\label{tab:palindromesab} 
\end{center} 
 \end{table}
\end{center} 

\begin{center}
\begin{table}[!ht] 
\begin{center}
\begin{tabular}{||c|c|c|c||} 
\hline \hline
Case of $n$		 	&$n=3$ 				&$n=4$			 	\\
\hline
 $k$				 &$2^{m+1}-3$		&$2^m+3$	  \\
$b$				 &$2^m$			&$2^{m+2}-2$  	\\	
$x=s_m[b \cdots b+k-1].$   &
$0^{2^m-2} \cdot 1 \cdot 0^{2^m-2}$	 &
$0 \cdot  1 \cdot 0^{2^m - 1} \cdot  1 \cdot 0$	 	\\
 \hline \hline	
\end{tabular}
\caption
{The $n$-th largest palindrome length, $k,$ for $s_m,$ for $3 \le m \le 8,$  for  $n=3,4,$  with one possible starting position, $b,$ and the underlying palindrome pattern, $x.$ It is conjectured that the table holds for all $m.$ }
\label{tab:palindromescd} 
\end{center} 
 \end{table}
\end{center} 

\begin{center}
\begin{table}[!ht] 
\begin{center}
\begin{tabular}{||c|c|||} 
\hline \hline
Case of $n$		 	 		&$n=5$		\\
\hline
$k$				  			&$2^m+1$ \\
$b$				 			&$  3 \times 2^{m - 2} + (2^{m - 1} - 1)2^{m + 1} - 1$	\\	
$x=s_m[b \cdots b+k-1]$		&   
 $0^{2^{m-2}} \cdot 1^{2^{m-1}+1} \cdot 0^{2^{m-2}}	$	\\
 \hline \hline	
\end{tabular}
\caption
{The $n$-th largest palindrome length, $k,$ for $s_m,$ $3 \le m \le 8,$  for  $n=5,$  with one possible starting position, $b,$ and the underlying palindrome pattern, $x.$ It is conjectured that the table holds for all $m.$ }
\label{tab:palindromese} 
\end{center} 
 \end{table}
\end{center}

 \begin{proof} The proofs for all five cases use the 3-step proof method, are similar to the proof of Theorem \ref{the:runs}(e), and hence are omitted.

\end{proof}

\begin{remark} 
The lack of obvious generalization of the palindromes in Tables \ref{tab:palindromesab} - \ref{tab:palindromese} affords  an opportunity to contrast the respective strengths and weaknesses of automatic proofs and  other proof methods. Automatic proofs, while error free and quick,  require a specific conjecture, typically the result of exploratory analysis, to prove. 
Contrastively, other proof methods suggest important pattern classes
and attributes during the exploratory process; for example, in the current paper, the
3-step proof method suggests
exploring patterns in the binary representations. Thus the proof method, besides being used for proofs, is also beneficial for exploration.
\end{remark}

\section{Unbordered Factors}

We begin with a discussion of the local case.

\begin{lemma} For $3 \le m \le 8, k \ge 1,$ there is an unbordered factor of length $k$ in $s_m.$ \end{lemma}
\begin{proof} Standard Walnut code, \cite[Section 8.6.13]{Shallit}. \end{proof}

This naturally motivates the following conjecture.
\begin{conjecture} For every $m \ge 3, k \ge 1,$ there is an unbordered factor of $s_m$ of length $k.$ \end{conjecture} 

We then proceed as usual using the \textit{prove more to establish less} method. Motivated by results for small $m$ establishable
with Walnut, we have the following conjecture for one infinite subset of the integers.

\begin{conjecture}\label{con:borderpower2} Let $m \ge 3.$ Then  \\
(a) For odd $m, e \ge 0,$ $s_m[0..2^{m+e}-1]$ is an unbordered factor of $s_m$.\\
(b) For even $m,e \ge 0,$ $s_m[2^m-1..(2^m-1)+(2^{m+e}-1)]$ is an unbordered factor of $s_m$.\\
\end{conjecture}.

However, we were unable to prove these conjectures despite the abundance of patterns. For example, for  $m$ odd, we can completely describe with proof that  
$\{n < 2^{m+e+2},e \ge 0: s_m[n]=1\}= \{1^m_v, 1^m \cdot 0_v, 1 \cdot 0 \cdot 1^m_v, 1^m \cdot 0 \cdot 0_v, 1^m \cdot 0 \cdot 1_v, 1^{m+2}_v\}.$   The situation for the suffixes is more complicated but 
there are definite patterns of which we mention one: For $m = 3+e, n=2m+1,$ one suffix of $s_m[0..2^n-1]$ is
$1^m \cdot 0^{m+1}+ \{j: 0 \le j < 2^{m}-\epsilon(m)\}
\cup_{i=1}^{\frac{m-1}{2}} \biggl( 1^{m+2i} \cdot 0^{m+1-2i}+\{j: 0 \le j < 2^{m-2i}-1\} \biggr)
\cup \{1^n\},$
with $\epsilon(m) \in \{1,2\}$ (it appears that $\epsilon(m)=1, \text{ if } m \equiv \pm 1 \pmod{8} \text{ and  $2$ otherwise}.$)
   The attractive feature of this suffix is that the gaps between $p$ such that $s_m[p]=1$ is bounded by 1 with a finite number of exceptions with predictable properties.

Despite these patterns, no proof was found for   the assertions of Conjecture \ref{con:borderpower2}. One obstacle is that the prefix of length
$2^{m+e+2}, m \ge 3, m \text{ odd, }$ described above can repeat quite often (but it doesn't repeat right before the terminal suffix.) In Section \ref{sec:hierarchy} 
we attributed this difficulty in proof to 
the fact that statements that a factor is unbordered are $\Sigma_3.$

\section{Conclusion}

This paper has studied a simply-defined family of automatic sequences. By using a proof method based on 
a correspondence between binary strings under concatenation and integers under addition and multiplication, several 
interesting results were proven (and/or conjectured) about maximal runs, palindromes, squares, and borders. We therefore believe
that \textit{families of sequences} is a fruitful concept  extending
the automatic sequences similar to the regular and synchronized sequences. Supporting this viewpoint, we point out that
another well known family of sequences, the Sturmian sequences \cite{Pecan1}, has recently received attention including
automatic proof software based on B\"uchi automata \cite{Pecan3} similar to Walnut, that can prove many 
interesting results \cite{Pecan2}. It is even possible that this software can prove many of the conjectures in this paper
but no attempt was made to do this.

\end{document}